\documentclass[12pt]{article}
\usepackage{subcaption}
\usepackage{amsmath, amsfonts, amssymb, amsthm, graphicx, geometry, arydshln, umoline, newtxtext, newtxmath, color, comment, soul, enumerate, bm, physics, comment, tikz, hyperref}
\usepackage{appendix}
\usepackage{natbib}
\usepackage{booktabs}
\usepackage[table]{xcolor}
\def\BibTeX{{\rm B\kern-.05em{\sc i\kern-.025em b}\kern-.08em
    T\kern-.1667em\lower.7ex\hbox{E}\kern-.125emX}}
\usepackage{makecell}
\usepackage{bbm}

\allowdisplaybreaks

\theoremstyle{definition}
\newtheorem{theorem}{Theorem}

\newtheorem{corollary}{Corollary}

\newtheorem{definition}{Definition}
\newtheorem{example}{Example}

\newtheorem{lemma}{Lemma}

\newtheorem{proposition}{Proposition}

\newcommand{\argmax}{\mathop{\rm arg~max}\limits}

\title{\huge{The Coarse Nash Bargaining Solutions}\thanks{The authors are grateful to  Koichi Tadenuma for their helpful comments  and suggestions from the very early stages of this research.
We also thank Walter Bossert, Youngsub Chun, Youichiro Higashi, Toru Hokari, Kohei Kamaga, Noriaki Kiguchi, and Kaname Miyagishima, as well as seminar participants at Okayama University and at the Waseda Workshop on Game Theory and Experiment 2025, Sophia University, and the 31st Decentralization Conference in Japan for their helpful comments.
We gratefully acknowledge the financial support from the Japan Society
for the Promotion of Science KAKENHI, Nos. 25K16606 (Nakada) and 25KJ1298 (Nakamura).}}
\author{Satoshi Nakada\thanks{School of Management, Department of Business Economics, Tokyo University of Science, 1-11-2, Fujimi, Chiyoda-
ku, Tokyo, 102-0071, Japan. Email: snakada@rs.tus.ac.jp} 
\and
Kensei Nakamura\thanks{Graduate School of Economics, Hitotsubashi University, 2-1, Naka, Kunitachi, Tokyo 186-8601, Japan. E-mail: kensei.nakamura.econ@gmail.com}
}
\date{This version: \today}

\begin{document}

\maketitle


\begin{abstract}
This paper studies the axiomatic bargaining problem and proposes a new class of bargaining solutions, called \textit{coarse Nash solutions}. 
These solutions assign to each problem a set of outcomes coarser than that chosen by the classical Nash solution (\citealp{nash1950bargaining}). Our main result shows that these solutions can be characterized by new rationality axioms for choice correspondences, which are modifications of Nash's independence of irrelevant alternatives (or more precisely, \citeauthor{arrow1959rational}'s \citeyearpar{arrow1959rational} choice axiom), when combined with standard axioms. 
\vspace{5mm}
\end{abstract}

\section{Introduction}

Consider a group of individuals facing a collective choice problem, in which each individual only cares about the payoff of the chosen outcomes. 
In a seminal paper, \citet{nash1950bargaining} formalized these situations as bargaining problems (sets of feasible utility vectors) and studied mappings $F$ that indicate the chosen outcomes, called bargaining solutions. 
Nash introduced several desirable axioms for bargaining solutions and showed that the Nash solution, which selects alternatives that maximize the product of the individuals’ utility levels of the outcomes for each feasible set, is the unique solution that satisfies all of them.

A key axiom for the Nash solution is \citeauthor{arrow1959rational}’s (\citeyear{arrow1959rational}) choice axiom, a set-valued version of Nash’s independence of irrelevant alternatives (cf. \citealp{kaneko1980extension,xu2006alternative}).
This axiom states that if the feasible set expands and the rule still chooses some alternatives from the original set, then all originally chosen alternatives must still be selected (i.e., for two problems $S, S'$ such that $S \subseteq S'$, $F(S')\cap S\neq \emptyset$ implies $F(S) = F(S')\cap S$).
This property is closely related to the rationalizability of multi-valued bargaining solutions or choice correspondences via welfare functions or orderings \citep{ok1999revealed}.

Arrow’s axiom has been criticized from various perspectives, leading to the proposal of weaker rationality axioms for choice correspondences and alternative solution concepts. 
It is well-known that Arrow’s axiom can be decomposed into two weaker axioms: the Chernoff axiom and the dual Chernoff axiom. 
Suppose that a choice set $S$ expands to $S'$ and the rule still selects some outcomes from $S$. 
The Chernoff axiom requires that if the chosen outcomes in $S'$ are feasible in $S$, then they must also be chosen in $S$ (i.e., $F(S')\cap S\neq \emptyset$ implies $F(S) \supseteq F(S')\cap S$). 
The dual Chernoff axiom requires that whenever some outcome chosen in $S$ is selected in $S'$, all outcomes chosen in $S$ must also be done so in $S'$ (i.e., $F(S')\cap S\neq \emptyset$ implies $F(S) \subseteq F(S')\cap S$). 
Most criticisms of Arrow's axiom target the Chernoff axiom.
In contrast, we argue that the dual Chernoff axiom is also subject to critique.\footnote{
For instance, \citet{Thomson1994CHAP} questioned Nash's independence of irrelevant alternatives under the name of contraction independence as follows:  ``(I)f the contraction described in the hypothesis of contraction independence is skewed against a particular agent, why should the compromise be prevented from moving against him?'' This criticism is related to whether the outcome chosen in the original problem should remain chosen in a smaller problem, and this relates to the Chernoff axiom. Similar arguments have been made by many other studies, such as \citet{roth1979axiomatic} and \citet{luce2012games}. Motivated by these criticisms, various alternative solution concepts, such as the Kalai-Smorodinsky solution \citep{kalai1975other}, have been introduced. 
}

To illustrate, suppose that $a$ and $b$ are available in $S$ and that both are chosen because they are incomparable. 
If $c$ becomes available in the expanded problem $S' = S \cup \{c\}$ and $c$ is preferred to $a$ but incomparable to $b$, then choosing $b$ and $c$ in $S'$ but not $a$ seems reasonable. 
However, the dual Chernoff axiom rules out such a decision because it requires that if $b$ is chosen, then $a$ should also be chosen.

This condition is naturally justified when an explicit objective function exists, as in the Nash solution, because the chosen outcomes coincide with those that dominate the remaining outcomes in terms of the objective function. However, constructing a single objective function becomes difficult when multiple criteria must be considered. 
For example, difficulties in aggregating diverse value judgments, such as efficiency and equity, or attempting to determine the bargaining power of individuals in real-life negotiations often resist reduction to a single objective function. 
In such situations, the dominance relation itself may remain incomplete, and the solution concept should allow for greater flexibility or a certain ``buffer.''

This observation raises a natural question: What rationality axioms should we impose to construct collective choice rules, and what class of rules satisfies them? 
To address this issue, we introduce two axioms by weakening Arrow's choice axiom.
The first axiom modifies the dual Chernoff axiom while retaining the Chernoff axiom.
Recall that the original dual Chernoff axiom implies that if an outcome is chosen in both $S$ and $T$, it should also be chosen in $S \cup T$. 
Our axiom also requires the converse: If an outcome is chosen in $S \cup T$ and is feasible in both $S$ and $T$, it must be chosen in both. 
We call this axiom the \textit{weak Arrow axiom} because it coincides with the combination of the Chernoff axiom and the weak dual Chernoff axiom known in the literature.
The second axiom postulates that all outcomes chosen in the original problem should be selected in an expanded problem if no outcome that becomes feasible owing to the expansion is chosen.  
This prevents the disappearance of all chosen outcomes upon irrelevant expansions; we therefore call this axiom \textit{independence of irrelevant expansions}. 
We investigate implications of these axioms under standard axioms such as efficiency, scale invariance, and continuity.

Our main result shows that, under these axioms, the dominance relation induced by a solution can be generalized. 
Recall that the dominance relation in the Nash solution is represented by comparisons of the products of individuals' utility levels.
Equivalently, $x$ dominates $y$ if the log-transformed difference in their utility vectors lies in the set $\{ z \in \mathbb{R}^n \mid \sum_{i = 1}^n z_i > 0 \}$. 
Our representation theorem generalizes this set to any open set $A$, which we call an \textit{\textbf{improving set}}, satisfying two conditions: (i) $A$ contains all strictly positive vectors and no weakly negative vectors (monotonicity), and (ii) for all $x, y \in A$, $x + y \in A$ (two-stage improvement property).

This class of solutions includes several meaningful examples.
At one extreme are the weighted Nash solutions, while at the other extreme is the solution that assigns all weakly Pareto optimal outcomes.
If $A$ is a convex cone, the dominance relation corresponds to the unanimity rule over multiple weighted Nash solutions. 
Since these solutions in the main result are less decisive than the original Nash solution but still based on the comparison of the Nash product, we call them \textbf{\textit{coarse Nash solutions}}.

Furthermore, we formally show that these solutions yield larger choice sets than weighted Nash solutions: More precisely, for every coarse Nash solution, there exists a weight vector such that its chosen set always includes the outcomes chosen by the corresponding weighted Nash solution. Thus, each coarse Nash solution is literally coarser than some weighted Nash solution.

This paper is organized as follows. Section \ref{sec_setup} formalizes the bargaining problem and the basic axioms satisfied by the Nash solution. Section \ref{sec_rat} introduces axioms motivated by the limitations of Arrow’s choice axiom. Section \ref{sec_main} examines the implications of these axioms under standard conditions and shows that any solution satisfying them can be represented as a coarse Nash solution. 
Section \ref{sec_literature} discusses the related literature. 
Finally, Section \ref{sec_con} concludes the paper.

\section{The bargaining problem}\label{sec_setup}

Let $N = \{ 1,2,\ldots, n \}$ be the fixed set of players, where $n\geq 2$. 
If they reach a unanimous agreement, then they attain a utility vector $x = (x_1, x_2, \ldots, x_n) \in \mathbb{R}_{++}^n$, where $x_i$ is player $i$'s utility level.\footnote{Let $\mathbb{R}$ (resp. $\mathbb{R}_{+}$, $ \mathbb{R}_{++}, $ and $\mathbb{R}_{-}$) denote the set of real numbers (resp. nonnegative numbers, positive numbers, and nonpositive numbers). 
Let $\mathbb{R}^n$ (resp. $\mathbb{R}^n_{+}$, $ \mathbb{R}^n_{++}$, and $ \mathbb{R}^n_{-}$) denote the $n$-fold Cartesian product of $\mathbb{R}$ (resp. $\mathbb{R}_{+}$, $ \mathbb{R}_{++}$, and $ \mathbb{R}_{-} $). Let $\mathbb{N}$ be the set of natural numbers and $\mathbb{Q}$ be the set of rational numbers.} 
If they fail to reach an agreement, then they stay at the origin $\mathbf{0} = (0,0,\ldots, 0) \in \mathbb{R}^n$. 
Here, we assume that all agreements give the players positive utility levels. 
A \textit{(bargaining) problem} $S\subseteq \mathbb{R}^n_{++}$ is a set of utility vectors attainable through unanimous agreements. 
We assume that each problem $S$ is compact relative to $\mathbb{R}^n_{++}$ and comprehensive (i.e., for all $x,y\in \mathbb{R}^n_{++}$ with $x\geq y$, $x\in S$ implies $y\in S$).%
\footnote{For $x,y\in \mathbb{R}^n$, we write $x \gg y$ if $x_i > y_i$ for all $i\in N$, and $x \geq y$ if $x_i \geq y_i$ for all $i\in N$. }  
Let $\mathcal{B}$ be the set of problems. 
For any $A \subseteq \mathbb{R}_{++}$, let
\[
\text{cmp}A=\{y \in \mathbb{R}^n_{++} \mid \text{there exists}~x \in A~\text{such that}~x\ge y\}
\]
be the comprehensive hull of $A$, that is, the smallest comprehensive set including $A$ (within $\mathbb{R}_{++}$).

We introduce several notations for convenience. A permutation is a bijection on $N$. Let $\Pi$ be the set of permutations. For $x\in \mathbb{R}^n $ and $\pi \in \Pi$, define $x_\pi$ as the vector $(x_{\pi (1)}, x_{\pi (2)}, \ldots , x_{\pi (n)} ) \in \mathbb{R}^n$. 
A set $S \subseteq \mathbb{R}^n $ is said to be \textit{symmetric} if $S = \{ x_\pi \mid x\in S \}$ for all $\pi\in \Pi$.   
For $a, x\in \mathbb{R}^n$, let $a\ast x = (a_1x_1, a_2x_2, \ldots, a_n x_n)$. 
Similarly, for $a\in \mathbb{R}^n$ and $S\subseteq \mathbb{R}^n$, let $a\ast S = \{ a\ast x \mid x\in S \}$.

A \textit{(bargaining) solution} $F$ assigns a nonempty subset $F (S)$ of $S$ to each problem $S\in \mathcal{B}$. 
Let $\Delta_N \subseteq \mathbb{R}^n_+$ be the set of weights over the players. 
For $w\in \Delta_N$, the \textbf{weighted Nash solution associated with} $w$, denoted by $F^w$, is the solution such that for all $S\in \mathcal{B}$, 
\begin{equation}
    F^w (S) = \argmax_{x\in S} \prod_{i\in N} x_i^{w_i}. 
\end{equation}
A special case with an equal-weight vector (i.e., $w_i=w_j$ for all $i,j\in N$) is the \textbf{Nash solution} $F^\text{N}$.
That is, for all $S\in \mathcal{B}$, 
\begin{equation}
    F^\text{N}(S) = \argmax_{x\in S} \prod_{i\in N} x_i. 
\end{equation}
In other words, the Nash solution chooses the maximizers of the product of the player's utility levels for each problem. 
The following are standard axioms and their variants that the Nash solution satisfies in our setup.

\begin{description}
    \item[\bf Arrow Axiom.] For all $S, S'\in \mathcal{B}$ with $S\subseteq S'$, if $F(S') \cap S \neq \emptyset$, then $F(S) = F(S')\cap S$. 

    \item[\bf Efficiency.] For all $S\in \mathcal{B} $ and $x\in F(S)$, there is no $y \in S$ such that $y \gg x$. 
    Furthermore, for all $x'\in S$, if $x'\geq y'$ for some $y'\in F(S)$, then $x'\in F(S)$. 

    \item[\bf Anonymity.] For all symmetric problems $S\in \mathcal{B}$ and $x\in S$, if $x\in F(S)$, then $x_\pi \in F(S)$ for all $\pi \in \Pi$. 

    \item[\bf Scale Invariance.] For all $S\in \mathcal{B}$, $a\in\mathbb{R}^n_{++}$, and $x\in S$, we have $x\in F(S) $ if and only if $ a \ast x\in F(a \ast S)$. 

    \item[\bf Continuity.] For all $S\in \mathcal{B}$ and $x\in S$, if there exist $\{S^k \}_{k\in \mathbb{N}} \subseteq \mathcal{B}$ and $\{ x^k\}_{k\in \mathbb{N}}\subseteq \mathbb{R}_{++}^n$ such that (i)  $x^k \in F(S^k)$ for all $k\in \mathbb{N}$; (ii) $\{S^k \}_{k\in \mathbb{N}}$ converges to $S$ in the Hausdorff topology; and (iii) $\{ x^k\}_{k\in \mathbb{N}}$ converges to $x$ in $\mathbb{R}_{++}^n$, then $x\in F(S)$.
\end{description}

\textit{Arrow axiom} \citep{arrow1959rational} is a set-valued version of Nash's \textit{independence of irrelevant alternatives (IIA)}. 
Under the usual interpretation, this axiom requires that if a feasible set becomes smaller, then the set of chosen outcomes should equal the restriction of the originally chosen set. 
In \textit{efficiency}, we require that any chosen outcome should not be strictly Pareto dominated, and that if an outcome is better than some chosen outcome for every player, then it should also be chosen. 
This ensures that 
at least one (strictly) Pareto efficient outcome is chosen.
Interpretations of the remaining axioms are standard.

\section{Rationalizabiliy}
\label{sec_rat}
In the classical characterization of the Nash solution, IIA plays a crucial role and provokes enormous discussions about the Nash solution. 
For the set-valued version of the Nash solution, \textit{Arrow axiom} is a natural counterpart and plays the same crucial role \citep{kaneko1980extension,xu2006alternative}.
It is known that this axiom can be decomposed into the following two parts:

 \begin{description}
    \item[\bf Chernoff Axiom.] For all $S, S'\in \mathcal{B}$ with $S\subseteq S'$, if $F(S') \cap S \neq \emptyset$, then $F(S) \supseteq  F(S') \cap S$.

 \item[\bf Dual Chernoff Axiom.] For all $S, S'\in \mathcal{B}$ with $S\subseteq S'$, if $F(S')\cap S \neq \emptyset$, then $F(S) \subseteq  F(S')\cap S$.
\end{description}

Based on this decomposition, the requirements of \textit{Arrow axiom} can be understood as follows.
Suppose that a choice set $S$ expands to $S'$ and the rule still chooses some outcomes from the original problem $S$ (that is, $F(S') \cap S\neq \emptyset$).  Then, the first condition requires that the choice from the expanded problem $S'$ that was available in the original problem should only be taken from the original chosen outcomes (i.e., \textit{Chernoff axiom}); and the second condition requires that all the chosen outcomes in the original problem $S$ must remain chosen (i.e., \textit{dual Chernoff axiom}).\footnote{The Chernoff axiom and Dual Chernoff axiom are also called Sen's $\alpha$ and $\beta$, respectively \citep{sen1970collective}.}


In the literature, the first requirement has been criticized, and several alternative axioms have been introduced to propose more appropriate solution concepts.
Here, however, we focus on the second condition and argue that it is not free from criticisms, too.

To illustrate this, we consider the following situations.
Suppose that $a$ and $b$ are available in a choice set $S$ and that both options are chosen because they are incomparable. 
Now, imagine that another option $c$ becomes available by expanding the choice set to the problem $S'$, and $c$ is more appealing than $a$.
We also assume that there is still no clear relationship between $b$ and $c$.
Then, choosing $b$ and $c$ but not $a$ in $S$ would be a natural decision.
However, \textit{dual Chernoff axiom} does exclude such choices. 

In a social decision context, this intuition requires us to agree with the following scenarios: Suppose that outcome $E$ is equitable but offers relatively low utility levels to individuals, whereas outcome $U$ is unequal but provides a sufficiently high total utility. 
One problem with taking $U$ is that the redistribution of utility among individuals is not easily achieved due to political or ideological reasons.
Now, suppose that a more appealing option $E'$ than $E$ becomes feasible.
For example, $E'$ is an outcome in which the utility of every individual is higher than that under $E$, though its total utility is still lower than that under $U$.
In such a situation, our common intuition suggests that $U$ and $E'$ should be selected from the expanded choice set and $E$ should no longer be chosen. 
However, if we accept \textit{Arrow axiom}, particularly \textit{dual Chernoff axiom}, such a decision must be ruled out.

To avoid these counterintuitive conclusions stemming from \textit{dual Chernoff axiom}, we consider the following weaker requirements:

\begin{description}
    \item[Weak Dual Chernoff Axiom.] For all $S, T\in\mathcal{B}$, $F(S)\cap F(T) \subseteq F(S \cup T)$.\footnote{This axiom is also called Sen's $\gamma$ \citep{sen1971choice}, which is sometimes written as follows: For any $\{S^k\}_{k=1}^K \subseteq \mathcal{B}$, $\bigcap_{k=1}F(S^k) \subseteq F(\bigcup_{k=1}^KS^k)$. By induction, the equivalence of both axioms can be shown, so that we can use the present axiom without loss of generality.} 
\end{description}

\begin{description}
    \item[\bf Independence of Irrelevant Expansions (IIE).] For all $S, S'\in \mathcal{B}$ with $S\subseteq S'$, if $F(S')\subseteq S$, then $F(S) \subseteq  F(S')$. 
\end{description}

\textit{Weak dual Chernoff axiom} is already known in the literature on social choice theory. 
It requires that, for any problem $S$ and $T$, the commonly chosen outcomes in both problems should also be chosen in the combined problem $S \cup T$.
This axiom is indeed weaker than \textit{dual Chernoff axiom} because the latter implies both $F(S)\subseteq F(S\cup T)$ and $F(T)\subseteq F(S\cup T)$.
Note that this axiom can accommodate the intuitive decision in the above scenarios.
\textit{IIE} is an axiom about problem expansions.\footnote{\citet{panda1983non} calls this axiom Sen's $\delta^*$, which is a stronger version of Sen's $\delta$ \citep{sen1971choice}. \citet{cato2018choice} calls this axiom weak Nash axiom $\beta$.} 
It requires that if the expansion is irrelevant in the sense that all the outcomes chosen in the expanded problem $S'$ can be taken in the original problem $S$, then the outcomes chosen in $S$ should be chosen in $S'$ again. 
This axiom avoids the above criticism of \textit{Arrow axiom} by strengthening the prerequisite of \textit{dual Chernoff axiom}. 

We adopt \textit{Chernoff axiom}, \textit{weak dual Chernoff axiom}, and \textit{IIE}, and examine their implications under the aforementioned standard axioms.
We first show that the implications of the combination of \textit{Chernoff axiom} and \textit{weak dual Chernoff axiom} are summarized as follows:

\begin{proposition}\label{prop_weakArrow}
A bargaining solution satisfies \textit{Chernoff axiom} and \textit{weak dual Chernoff axiom} if and only if it satisfies the following property:
\begin{description}
\item[]\centering  For all $S, T\in\mathcal{B}$, ~ $F(S)\cap F(T) = F(S \cup T) \cap S\cap T$.
\end{description}
\end{proposition}

We call the new property derived in Proposition \ref{prop_weakArrow} \textbf{\textit{weak Arrow axiom}}. 
As discussed above, we regard the combination of \textit{weak Arrow axiom} and \textit{IIE} as a weakened counterpart of \textit{Arrow axiom} and consider implications of the axioms for bargaining solutions.
For a binary relation $\succsim$ over $D\in \{\mathbb{R}^n_{++}, \mathbb{R}^n\}$, let $\succ$ denote its asymmetric part. 
A binary relation $\succsim$ is said to be \textit{quasi-transitive} if for all $x,y,z,\in D$, $x\succ y$ and $y\succ z$ imply $x\succ z$; \textit{monotone} if for all $x,y\in D$, $x\succ y$ whenever $x \gg y$; and \textit{continuous} if for each $x\in D$, $\{ y\in D \mid y\succ x\}$ is open. 
\textit{Weak Arrow axiom} and \textit{IIE} have the following implications in our context.

\begin{lemma}\label{prop_rationalization}
Suppose that a solution $F$ satisfies \textit{efficiency} and \textit{continuity}.
Then, the following two statements hold.

\begin{itemize}
\item[(1)] $F$ satisfies 
\textit{Arrow axiom} if and only if it is \textit{rationalizable}, that is, there exists a complete, transitive, monotone, and continuous binary relation $\succsim$ such that
\[
F(S)=\{x \in S\mid x \succsim y, \forall y \in S\},~~\forall S \in \mathcal{B}.
\]

\item[(2)] $F$ satisfies \textit{weak Arrow axiom} and \textit{IIE} if and only if it is \textit{weakly rationalizable}, that is, there exists a quasi-transitive, monotone, and continuous binary relation $\succsim$ such that
\[
F(S)=\{x \in S\mid \not \exists y \in S ~~\text{s.t.} ~~ y\succ x\},~~\forall S \in \mathcal{B}.
\]

\end{itemize}

\end{lemma}

In the finite-alternative case, \citet{arrow1959rational} established the well-known characterization of rationalizability by \textit{Arrow axiom}, and \citet{tyson2008cognitive} provided the corresponding characterization of weak rationalizability by axioms corresponding to \textit{weak Arrow axiom} and \textit{IIE} (see also Sen (1971) for another characterization).
The first statement in Lemma \ref{prop_rationalization} is a minor modification of \citeauthor{ok1999revealed}'s (\citeyear{ok1999revealed}) extension of Arrow’s result from the finite setting to the non-convex bargaining problem. 
In parallel, the second statement extends \citeauthor{tyson2008cognitive}’s (\citeyear{tyson2008cognitive}) characterization of weak rationalizability to the non-convex bargaining domain. 
As in \citet{ok1999revealed}, the additional axioms of \textit{efficiency} and \textit{continuity} are imposed to address the complexity stemming from the difference in domains.


\section{The coarse Nash solutions}
\label{sec_main}

This section examines the implications of the axioms introduced in the previous section by providing a representation theorem for them. 

Before considering \textit{weak Arrow axiom} and \textit{IIE}, we begin with the relationship between Lemma \ref{prop_rationalization} (1) and the Nash solution. 
Lemma \ref{prop_rationalization} (1) states that if a solution satisfies \textit{Arrow axiom} and several regularity axioms, it can be rationalized by a complete, transitive, monotone, and continuous binary relation $\succsim$ over $\mathbb{R}^n_{++}$.
We now consider the binary relation $\succsim^*$ over $\mathbb{R}^n$ defined by 
\begin{equation*}
    x\succsim^* y \iff (e^{x_1}, \ldots, e^{x_n}) \succsim (e^{y_1}, \ldots, e^{y_n}). 
\end{equation*}
In other words, for all $x,y\in \mathbb{R}^n_{++}$, we have 
\[
\log x \succsim^*\log y \iff x \succsim y, 
\]
where we denote $(\log x_1, \ldots, \log x_n)$ and $(\log y_1, \ldots, \log y_n)$ by $\log x$ and $\log y$, respectively. 
By construction, $\succsim^*$ also satisfies completeness, transitivity, monotonicity, and continuity.
If we impose \textit{scale invariance} on $F$, we can see that $\succsim^*$ satisfies additivity.\footnote{Formally, for any $\log x, \log y, a \in \mathbb{R}^n$, $\log x \succsim^* \log y \Leftrightarrow \log x+a \succsim^* \log y+a$.}
Therefore, de Finetti's theorem (cf. Chapter 9 in \citealp{gilboa2009theory}) implies that there exists $w\in \Delta_N$ such that for all $x,y\in \mathbb{R}^n_{++}$,  
\[
\log x \succsim^* \log y \iff (\log x-\log y) \succsim^* \textbf{0} \iff \sum_{i \in N} w_i (\log x_i- \log y_i) \ge 0.
\]
That is, for any $S \in \mathcal{B}$, $x \in F(S)$ if and only if there is no $y\in S$ such that  
\begin{equation}
\label{eq:nash_alt}
    \log y -\log x \in \qty{z \in \mathbb{R}^n \mid  \sum_{i \in N} w_i z_i \ge 0 }
\end{equation} 
(Figure \ref{fig:half}). 
Note that this is an alternative representation of the weighted Nash solution $F^w$, as we can check immediately. 
Here, if we regard a solution as an arbitrator's choice rule, then \eqref{eq:nash_alt} can be interpreted as follows: 
The benefit of the change from $x$ to $y$ is measured by the difference between their log-transformed vectors, and this change is regarded as a social improvement only if the difference is in the set  $\{z \in \mathbb{R}^n \mid \sum_{i \in N} w_i z_i \ge 0\}$. 
The weighted Nash solution $F^w$ assigns to each problem the set of vectors that cannot be improved in this sense. 
If we consider a solution as the prediction of bargaining outcomes, whether the difference in the log-transformed vectors is in the set $\{z \in \mathbb{R}^n \mid \sum_{i \in N} w_i z_i \ge 0\}$ can be a criterion for determining which vector is more appropriate as a prediction. 

\begin{figure}
    \centering
    \tikzset{every picture/.style={line width=0.75pt}}        
    \tikzset{every picture/.style={line width=0.75pt}}        

    \begin{tikzpicture}[x=0.75pt,y=0.75pt,yscale=-1,xscale=1]
    
    \draw (211,159) -- (450.5,159.12);
    \draw [shift={(450.5,159.12)}, rotate = 180.03] (10.93,-3.29) .. controls (6.95,-1.4) and (3.31,-0.3) .. (0,0) .. controls (3.31,0.3) and (6.95,1.4) .. (10.93,3.29);
    
    \draw (330,275.12) -- (330,38.12);
    \draw [shift={(330,38.12)}, rotate = 90] (10.93,-3.29) .. controls (6.95,-1.4) and (3.31,-0.3) .. (0,0) .. controls (3.31,0.3) and (6.95,1.4) .. (10.93,3.29);
    
    \draw (268.02,37.93) -- (391,277);
    
    \path[fill=gray!50, fill opacity=0.4] 
        (391,277) -- (450.44,277.02) -- (450.57,38.34) -- (268.02,38.31) -- cycle;
    
    \draw (315,163.4) node [anchor=north west][inner sep=0.75pt] {$0$};
    \draw (346,171.4) node [anchor=north west][inner sep=0.75pt] [font=\fontsize{0.59em}{0.71em}\selectfont] {$\left\{z\in \mathbb{R}_{++}^{n} \mid \sum _{i\in N} w_{i} z_{i} >0\right\}$};
    \draw (337,42.4) node [anchor=north west][inner sep=0.75pt] {$z_{2}$};
    \draw (435,140) node [anchor=north west][inner sep=0.75pt] {$z_{1}$};
    
    \end{tikzpicture}

    \caption{The set $\left\{z\in \mathbb{R}_{++}^{n} \mid \sum _{i\in N} w_{i} z_{i} >0\right\}$ for some $w\in \Delta_N$}
    \label{fig:half}
\end{figure}
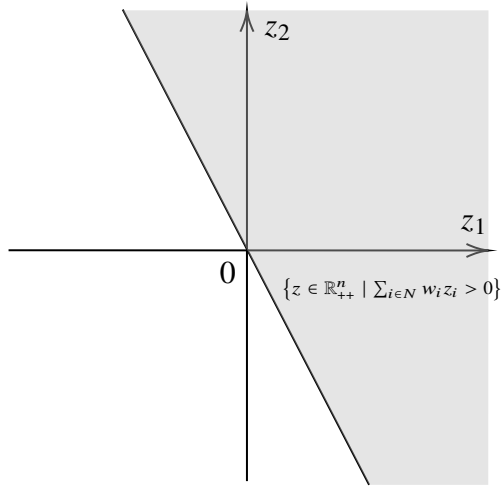

Conceptually, the right-hand side of \eqref{eq:nash_alt} need not be the set $\{z \in \mathbb{R}^n \mid \sum_{i \in N}w_i z_i \ge 0\}$; we can take other sets to determine whether a social improvement is achieved or not. 
Our main result shows that if we replace \textit{Arrow axiom} with \textit{weak Arrow axiom} and \textit{IIE}, then this set can be generalized, and the Nash solution can be modified to accommodate choice with some ``buffer.''

To provide a formal definition of these solutions, we introduce the following new concept. 
We say that a set $A \subseteq \mathbb{R}^n$ is an \textbf{\textit{improving set}} if $A$ is an open set such that  (i) $\mathbb{R}^n_{++} \subseteq A \subseteq \mathbb{R}^n \backslash \mathbb{R}^n_{-}$ and (ii) for all $x,y\in A$, $x + y \in A$.
The following is a new class of solutions parameterized by improving sets, which is the main focus of this paper. 

\begin{definition}
    For an improving set $A\subseteq \mathbb{R}^n$, a solution is a \textbf{\textit{coarse Nash solution associated with}} $A$, denoted by $F^A$, if for all $S\in \mathcal{B}$, 
    \begin{equation}
        F^A (S) = \{ x\in S \mid \nexists y \in S~~ \text{s.t. } \log y - \log x \in A \}.
    \end{equation}
\end{definition}

In the definition of coarse Nash solutions, the condition (i) of improving sets corresponds to the monotonicity of the dominance relation, while the condition (ii) means that two-stage improvements are also deemed as
desirable changes.
 
We then present our main result. 
The following theorem shows that the coarse Nash solutions characterize all the implications of \textit{weak Arrow axiom}, \textit{IIE}, and other basic axioms.

\begin{theorem}
\label{thm:main}
     A solution $F$ satisfies \textit{weak Arrow axiom}, \textit{IIE}, \textit{efficiency},  \textit{scale invariance},
     and \textit{continuity} if and only if it is a coarse Nash solution.
\end{theorem}

If \textit{anonymity} is additionally imposed, the improving set becomes symmetric, as stated in the following corollary: 

\begin{corollary}
\label{thm:main_anonymity}
     A solution $F$ satisfies \textit{weak Arrow axiom}, \textit{IIE}, \textit{efficiency},  \textit{scale invariance}, \textit{anonymity}, 
     and \textit{continuity} if and only if it is a coarse Nash solution associated with a symmetric improving set. 
\end{corollary}


To understand what solutions are included in the class of coarse Nash solutions, we provide several examples (and counterexamples) of them.

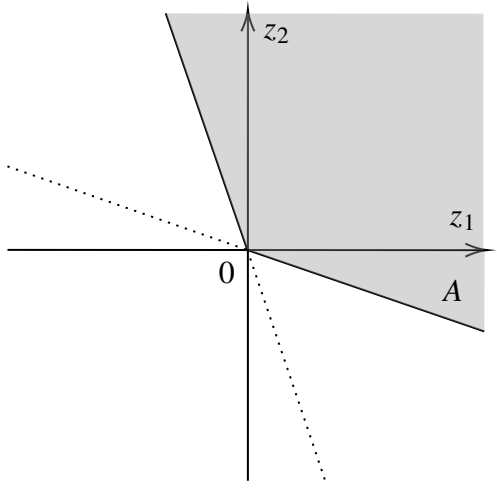
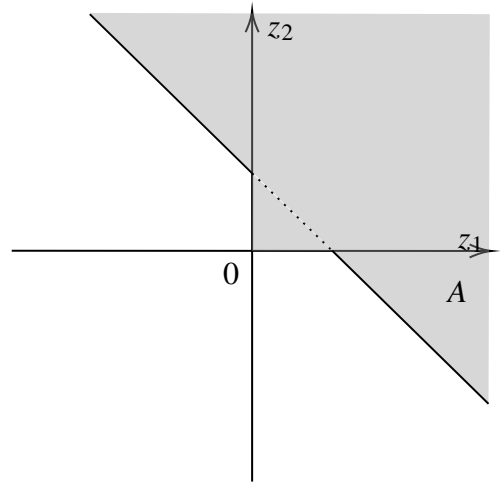
\begin{figure}
  \centering
  \begin{subfigure}[b]{0.48\textwidth}
    \centering

    \tikzset{every picture/.style={line width=0.75pt}} 
    
    \begin{tikzpicture}[x=0.75pt,y=0.75pt,yscale=-1,xscale=1]
    
    \draw    (211,159) -- (448.5,159.12) ;
    \draw [shift={(450.5,159.12)}, rotate = 180.03] [color={rgb, 255:red, 0; green, 0; blue, 0 }  ][line width=0.75]    (10.93,-3.29) .. controls (6.95,-1.4) and (3.31,-0.3) .. (0,0) .. controls (3.31,0.3) and (6.95,1.4) .. (10.93,3.29)   ;
    \draw    (331,275.12) -- (331,40.12) ;
    \draw [shift={(331,38.12)}, rotate = 90] [color={rgb, 255:red, 0; green, 0; blue, 0 }  ][line width=0.75]    (10.93,-3.29) .. controls (6.95,-1.4) and (3.31,-0.3) .. (0,0) .. controls (3.31,0.3) and (6.95,1.4) .. (10.93,3.29)   ;
    \draw    (290,40) -- (330.75,159.06) ;
    \draw  [color={rgb, 255:red, 0; green, 0; blue, 0 }  ,draw opacity=0 ][fill={rgb, 255:red, 155; green, 155; blue, 155 }  ,fill opacity=0.4 ] (330.75,159.06) -- (290,40) -- (331.13,40.13) -- cycle ;
    \draw  [color={rgb, 255:red, 0; green, 0; blue, 0 }  ,draw opacity=0 ][fill={rgb, 255:red, 155; green, 155; blue, 155 }  ,fill opacity=0.4 ] (330.81,40.13) -- (448.57,40.19) -- (448.51,159.12) -- (330.75,159.06) -- cycle ;
    \draw  [dash pattern={on 0.84pt off 2.51pt}]  (330.75,159.06) -- (370.25,277.06) ;
    \draw    (330.75,159.06) -- (449,200) ;
    \draw  [dash pattern={on 0.84pt off 2.51pt}]  (211,117) -- (330.75,159.06) ;
    \draw  [color={rgb, 255:red, 0; green, 0; blue, 0 }  ,draw opacity=0 ][fill={rgb, 255:red, 155; green, 155; blue, 155 }  ,fill opacity=0.4 ] (449,200) -- (330.75,159.06) -- (449.13,159.43) -- cycle ;
    
    \draw (315,163.4) node [anchor=north west][inner sep=0.75pt]    {$0$};
    \draw (337,44) node [anchor=north west][inner sep=0.75pt]  {$z_{2}$};
    \draw (430,138) node [anchor=north west][inner sep=0.75pt] {$z_{1}$};
    \draw (426,172.4) node [anchor=north west][inner sep=0.75pt]    {$A$};

    \end{tikzpicture}
    \caption{The set $A$ in Example 3}
    \label{fig:ex3}
  \end{subfigure}
  \hfill
  \begin{subfigure}[b]{0.48\textwidth}
    \centering

    \tikzset{every picture/.style={line width=0.75pt}} 
    
    \begin{tikzpicture}[x=0.75pt,y=0.75pt,yscale=-1,xscale=1]
    
    \draw    (211,159) -- (448.5,159.12) ;
    \draw [shift={(450.5,159.12)}, rotate = 180.03] [color={rgb, 255:red, 0; green, 0; blue, 0 }  ][line width=0.75]    (10.93,-3.29) .. controls (6.95,-1.4) and (3.31,-0.3) .. (0,0) .. controls (3.31,0.3) and (6.95,1.4) .. (10.93,3.29)   ;
    \draw    (331,275.12) -- (331,40.12) ;
    \draw [shift={(331,38.12)}, rotate = 90] [color={rgb, 255:red, 0; green, 0; blue, 0 }  ][line width=0.75]    (10.93,-3.29) .. controls (6.95,-1.4) and (3.31,-0.3) .. (0,0) .. controls (3.31,0.3) and (6.95,1.4) .. (10.93,3.29)   ;
    \draw  [color={rgb, 255:red, 0; green, 0; blue, 0 }  ,draw opacity=0 ][fill={rgb, 255:red, 155; green, 155; blue, 155 }  ,fill opacity=0.4 ] (331,120) -- (371,159.12) -- (330.95,159.06) -- cycle ;
    \draw    (250,40) -- (331,120) ;
    \draw  [dash pattern={on 0.84pt off 2.51pt}]  (331,120) -- (371,159) ;
    \draw  [color={rgb, 255:red, 0; green, 0; blue, 0 }  ,draw opacity=0 ][fill={rgb, 255:red, 155; green, 155; blue, 155 }  ,fill opacity=0.4 ] (449,236) -- (250,39.37) -- (449.62,39.99) -- cycle ;
    \draw    (371,159) -- (449,236) ;
    
    \draw (315,163.4) node [anchor=north west][inner sep=0.75pt]    {$0$};
    \draw (337,42.4) node [anchor=north west][inner sep=0.75pt]  {$z_{2}$};
    \draw (432,149) node [anchor=north west][inner sep=0.75pt] {$z_{1}$};
    \draw (426,172.4) node [anchor=north west][inner sep=0.75pt]    {$A$};

    \end{tikzpicture}
    \caption{The set $A$ in Example 4}
    \label{fig:ex4}
  \end{subfigure}

  \caption{Examples of improving sets}
  \label{fig:ex}
\end{figure}

\begin{example}
Suppose that $A = \{ z\in \mathbb{R}^n \mid \sum_{i\in N}w_i z_i  > 0 \}$ for some $w\in \Delta_N$.
As discussed earlier, the coarse Nash solution $F^A$ coincides with the weighted Nash solution $F^w$, which is an extreme case.
\end{example}

\begin{example}
For the other extreme case, suppose that $A$ is the smallest set, that is $A=\mathbb{R}_{++}^n$.
Then, for all $x,y\in \mathbb{R}_{++}^n $, 
\begin{align*}
    \log y - \log x \in A 
    &\iff [~  \log y_i  - \log x_i   > 0  ~~ \text{for all } i\in N ~ ]  \iff [~  y_i >  x_i   ~~ \text{for all } i\in N ~ ], 
\end{align*}
which implies that $F^A (S) = \{ x\in S \mid \nexists y \in S~~ \text{s.t. } y \gg x \}$ for all $S\in \mathcal{B}$. 
That is, $F$ is the solution that assigns to each problem the set of all weakly Pareto optimal agreements. 
\end{example}

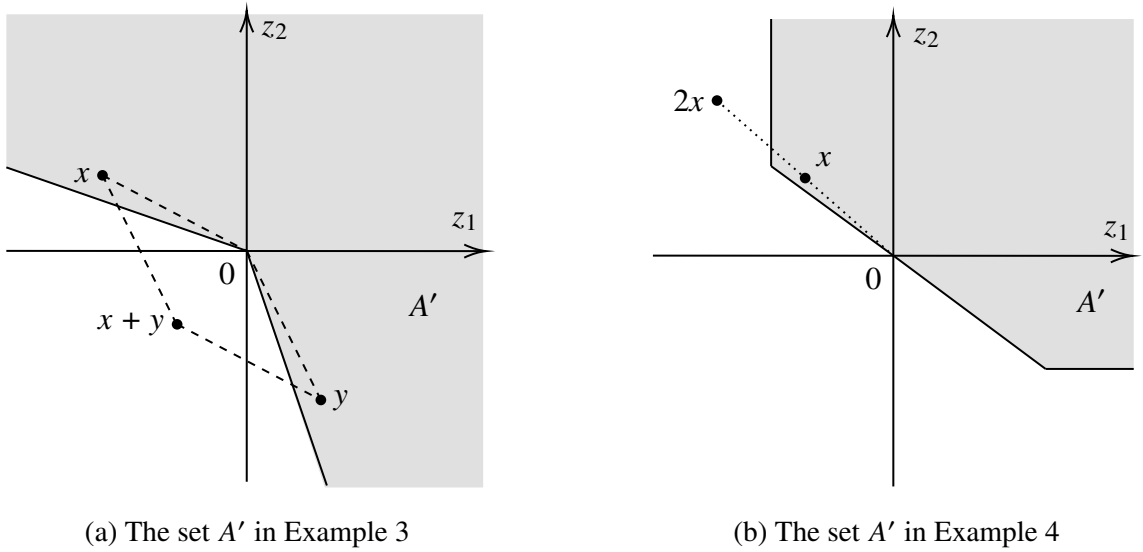
\begin{figure}
  \centering
  \begin{subfigure}[b]{0.48\textwidth}
    \centering
    \tikzset{every picture/.style={line width=0.75pt}}        

    \begin{tikzpicture}[x=0.75pt,y=0.75pt,yscale=-1,xscale=1]

    \fill[gray!60,opacity=0.4]
      (211,40) -- (449,40) -- (449,278) -- (370,278) -- (331,159) -- (211,117) -- cycle;
    
    \draw    (211,159) -- (450,159.12) ;
    \draw [shift={(450.5,159.12)}, rotate = 180.03] [color=black][line width=0.75]    
    (10.93,-3.29) .. controls (6.95,-1.4) and (3.31,-0.3) ..
    (0,0) .. controls (3.31,0.3) and (6.95,1.4) ..
    (10.93,3.29);
    
    \draw    (331,275.12) -- (331,38.12) ;
    \draw [shift={(331,38.12)}, rotate = 90] [color=black][line width=0.75]    
    (10.93,-3.29) .. controls (6.95,-1.4) and (3.31,-0.3) ..
    (0,0) .. controls (3.31,0.3) and (6.95,1.4) ..
    (10.93,3.29);
    
    \draw (211,117) -- (331,159);
    \draw (331,159) -- (371,277);
    
    \draw[dashed] (259,121) -- (331,159);
    \fill (259,121) circle (2pt) node[left] {$x$};
    
    \draw[dashed] (368,234) -- (331,159);
    \fill (368,234) circle (2pt) node[right] {$y$};
    
    \draw[dashed] (296,195.94) -- (259,121);
    \fill (296.25,195.94) circle (2pt) node[left] {$x+y$};
    
    \draw[dashed] (296.25,195.94) -- (368,234);
    
    \draw (331,159) node[below left] {$0$};
    
    \draw (345,48) node {$z_{2}$};
    \draw (440,144) node {$z_{1}$};
    \draw (419,186) node {$A'$};
    
    \end{tikzpicture}

    \caption{The set $A'$ in Example 3}
    \label{fig:cex3}
  \end{subfigure}
  \hfill
  \begin{subfigure}[b]{0.48\textwidth}
    \centering

    \tikzset{every picture/.style={line width=0.75pt}} 
    
    \begin{tikzpicture}[x=0.75pt,y=0.75pt,yscale=-1,xscale=1]
    \fill[gray!60,opacity=0.4]
      (270,40) -- (270,114) -- (331,159) -- (407,216) -- (451,216) -- (451,40) -- cycle;
    
    \draw (211,159) -- (448.5,159.12);
    \draw [shift={(450.5,159.12)}, rotate=180.03][color=black][line width=0.75]
      (10.93,-3.29) .. controls (6.95,-1.4) and (3.31,-0.3) .. (0,0)
      .. controls (3.31,0.3) and (6.95,1.4) .. (10.93,3.29);
    \draw (331,275.12) -- (331,40.12);
    \draw [shift={(331,38.12)}, rotate=90][color=black][line width=0.75]
      (10.93,-3.29) .. controls (6.95,-1.4) and (3.31,-0.3) .. (0,0)
      .. controls (3.31,0.3) and (6.95,1.4) .. (10.93,3.29);
    
    \draw (270,40) -- (270,114);        
    \draw (270,114) -- (331,159);       
    \draw (331,159) -- (407,216);       
    \draw (407,216) -- (451,216);       
    
    \draw[dotted] (243,81) -- (287,120) -- (331,159);
    \fill (287,120) circle (2pt) node[above right] {$x$};
    \fill (243,81)  circle (2pt) node[left]       {$2x$};
    
    \draw (331,159) node[below left] {$0$};
    \draw (335,39) node[anchor=north west] {$z_{2}$};
    \draw (430,136) node[anchor=north west] {$z_{1}$};
    \draw (416,170.4) node[anchor=north west] {$A'$};
    
    \end{tikzpicture}

    \caption{The set $A'$ in Example 4}
    \label{fig:cex4}
  \end{subfigure}

  \caption{Counterexamples of improving sets}
  \label{fig:cex}
\end{figure}

\begin{example}
\label{exm:2_multiweight}
Suppose that an improving set $A$ is the convex cone $\{ x\in \mathbb{R}^n_{++} \mid \sum_{i\in N} w_i x_i > 0 \quad (\forall w\in W ) \}$ for some nonempty closed subset $W \subseteq \Delta_N$ (Figure \hyperref[fig:ex3]{2(a)}). 
In this case, for all $x,y\in \mathbb{R}_{++}^n $, 
\begin{equation}
     \log y - \log x  \in A  \iff  \qty[ ~ \prod_{i\in N} y^{w_i}_i > \prod_{i\in N} x_i^{w_i} \quad  \forall w\in W ~], 
\end{equation}
which implies that $F(S) = \{ x\in S \mid \nexists y \in S ~~ \text{s.t. }  \prod_{i\in N} y^{w_i} > \prod_{i\in N} x^{w_i}_i ~~  (\forall w\in W) \}$ for all $S\in \mathcal{B}$. 
That is, the solution becomes a multi-weighted Nash solution, which chooses outcomes that can be justified by some set $W$ of weights.

By contrast, the set $ A' = \{ x\in \mathbb{R}^n_{++} \mid \exists w\in W ~~~\text{s.t.}~~~ \sum_{i\in N} w_i x_i > 0\}$ for some nonempty closed subset $W \subseteq \Delta_N$ with $|W|\geq 2$ is not an improving set (Figure \hyperref[fig:cex3]{3(a)}). 
For this set, we can take $x,y\in A'$ such that $x + y\in \mathbb{R}^n_-$, which violates  the definition of improving sets. 
Thus, the solution $F$ defined as $F(S) = \{ x\in S \mid \nexists y \in S ~~ \text{s.t. }  \prod_{i\in N} y^{w_i} > \prod_{i\in N} x^{w_i}_i ~~  (\exists w\in W) \}$ for all $S\in \mathcal{B}$ is not a coarse Nash solution.  
\end{example}

The final example shows that there is no logical relationship between improving sets and the convexity. 

\begin{example}
\label{exm:3_union}
Suppose $A =  \mathbb{R}^n_{++} \cup \{ x\in \mathbb{R}^n \mid \sum_{i\in N} x_i > \varepsilon \}$ for some $\varepsilon\in \mathbb{R}_{++}$ (Figure \hyperref[fig:ex4]{2(b)}). 
Note that this nonconvex set satisfies the definition of improving sets. 
Then, for all $x,y\in \mathbb{R}_{++}^n $, 
\begin{align}
    \label{eq:dom_ex}
     \log y - \log x \in A    
     \iff  [~  y_i >  x_i   ~~ \text{for all } i\in N ~ ] ~ \lor ~   \qty[ ~ \prod_{i\in N} y_i - \prod_{i\in N} x_i > \varepsilon ~].
\end{align}
Thus, $F^A$ can be represented as the solution that assigns to each problem $S\in \mathcal{B}$ the set of $x$ such that there is no $y \in S$ that satisfies \eqref{eq:dom_ex}. Note that the latter condition in \eqref{eq:dom_ex} can be interpreted as the dominance relation associated with the Nash solution with a threshold $\varepsilon$: 
According to this, $y$ is deemed to be better than $x$ if the difference in their Nash products is greater than $\varepsilon$. 

Furthermore, consider the convex set $A' \subseteq \mathbb{R}^n$  drawn in Figure \hyperref[fig:cex4]{3(b)}. 
Then, for some $x\in A'$, $x+x\notin A'$ holds, which violates the condition (ii) in the definition of improving sets. 
\end{example}

From these examples, one might think that every coarse Nash solution yields an improving set that is sharper than that of some weighted Nash solution. 
The following result shows this conjecture is true and justifies that we call them coarse Nash solutions: 
For any coarse Nash solution, there exists a weighted Nash solution whose chosen outcomes are always included within those of the coarse solution.

\begin{proposition}\label{prop_improving set}
For any improving set $A$, there exists $w \in \Delta_N$ such that $A \subseteq \{x \in \mathbb{R}^m \mid \sum_{i \in N}w_ix_i>0  \}$. 
Furthermore, for all $S\in \mathcal{B}$,  $ F^w (S) \subseteq F^A(S)$ holds. 
\end{proposition}

\begin{proof}
Let $A$ be an arbitrary improving set. 
Since $A$ is an open set, it is sufficient to prove that for some $w^*\in \Delta_N$, $\textbf{co}(A) \subseteq \{x \in \mathbb{R}^m \mid \sum_{i \in N}w^*_ix_i \geq 0  \}$ holds, where $\textbf{co}(A)$ is the convex hull of $A$. 
If $\mathbf{0}$ is not an interior point of $\textbf{co}(A)$, then by the separating hyperplane theorem, there exists $w^*\in \mathbb{R}^n$ such that $\sum_{i \in N}w^*_ix_i\geq 0 $ for all $x\in \textbf{co}(A)$. Since $\mathbb{R}^n_{+} \subseteq A$, we can set $w^*\in  \Delta_N$. 

To complete the proof, we show that $\mathbf{0}$ is not an interior point of $\textbf{co}(A)$. 
Suppose to the contrary that $\mathbf{0}$ is an interior point of $\textbf{co}(A)$.
Then, there exist $x^1, x^2, \ldots, x^K \in A$ and $\alpha^1, \alpha^2, \ldots, \alpha^K \in [0,1]$ with $\sum_{k = 1}^K \alpha^k = 1$ such that 
\begin{equation*}
    \sum_{k= 1}^K \alpha^k x^k \in \mathbb{R}^n_- .
\end{equation*} 
Note that we can take $\alpha^1, \alpha^2, \ldots, \alpha^K \in \mathbb{Q}\cap[0,1]$, that is, there exist $p^1, p^2, \ldots, p^K, q^1, q^2, \ldots, q^K \in \mathbb{N}$ such that $\alpha^k = p^k/q^k$ for each $k = 1,2,\ldots, K$. Let $x^* = (\prod_{k= 1}^K q_k) \sum_{k= 1}^K  p^k  x^k/q^k \in \mathbb{R}^n_-$. By applying the property (ii) in the definition of improving sets repeatedly,  we have $ x^* \in A$, which is a contradiction to the property (i) in the definition of improving sets. 
\end{proof}

To put it differently, this proposition states that 
for each coarse Nash solution $F^A$, there exists $w\in \Delta_N$ such that for all $x,y\in \mathbb{R}^n_{++}$, 
\begin{equation}
    x\succ y 
    \implies \prod_{i\in N} x_i^{w_i} \geq  \prod_{i\in N} y_i^{w_i}, 
\end{equation}
where $\succ$ denotes the asymmetric part of $\succsim$ in Proposition \ref{prop_rationalization}. 
The existence of such a weight is closely related to \citeauthor{aumann1962utility}'s (\citeyear{aumann1962utility}) result. 
Although we cannot directly compare these results because of the differences in frameworks, our proposition can be regarded as an extension of \citeauthor{aumann1962utility}'s (\citeyear{aumann1962utility}) result.  For a detailed discussion, see Section \ref{sec_literature}. 

Furthermore, from this result, we can obtain the following important properties of the collection of improving sets. 

\begin{corollary}\label{prop_improving set symmetric collection}
    For the collection of improving sets, $\mathbb{R}^n_{++}$ is the minimum set with respect to the set inclusion, and a set $A^*$ is a maximal set with respect to the set inclusion if and only if $A^* = \{x \in \mathbb{R}^m \mid \sum_{i \in N}w_ix_i>0  \}$ for some $w\in \Delta_N$. 
\end{corollary}

If we additionally require that an improving set $A$ be symmetric, a weight vector specified in Proposition \ref{prop_improving set} can be taken to be symmetric.
That is, under the symmetry of $A$, the solution $F^A$ assigns a coarser set of outcomes than that chosen by the Nash solution. 

\begin{proposition}\label{prop_improving set symmetric}
    If $A$ is a symmetric improving set, then $A\subseteq \{ x\in \mathbb{R}^n \mid \sum_{i\in N} x_i > 0 \}$. 
    Furthermore, for all $S\in \mathcal{B}$,  $ F^\text{N} (S) \subseteq F^A(S)$ holds. 
\end{proposition}

\begin{proof}
    Let $A$ be a symmetric improving set. Suppose to the contrary that there exists $x\in A$ with $\sum_{i\in N} x_i \leq 0$. Then, by the symmetry of $A$, for each $\pi\in \Pi$, we have $x^\pi \in A$. Let $x^\ast = \sum_{\pi \in \Pi} x^\pi$, that is, $x^\ast = \bigl( (n-1)!(\sum_{i\in N} x_i) , \ldots, (n-1)!(\sum_{i\in N} x_i) \bigr)$. By $\sum_{i\in N} x_i \leq 0$, we have $x^\ast\leq \mathbf{0}$. 
    By applying the property (ii) in the definition of improving sets repeatedly, we have $x^\ast \in A$, which is a contradiction.  
\end{proof}

\section{Related literature}\label{sec_literature}

Motivated by the problems of \textit{Arrow axiom} and Nash's IIA, many alternative axioms and solution concepts, such as \citeauthor{kalai1975other}'s  (\citeyear{kalai1975other}) solution, have been proposed (for comprehensive surveys, see \citet{thomson1981nash,Thomson1994CHAP,Thomson2022RED}).
Nash's IIA has been criticized for the property that the solution outcomes in the original problem must also be chosen in a smaller problem, even if the contraction dramatically changes the bargaining situations (e.g., \citet{luce2012games} and \citet{roth1979axiomatic}). 
Under the decomposition of \textit{Arrow axiom}, \textit{Chernoff axiom} has therefore attracted considerable attention. 
In contrast, this paper focuses on a problem that arises from \textit{dual Chernoff axiom}. 
By weakening this part, we have obtained a class of generalized Nash solutions called coarse Nash solutions. 

Relationships between the rationalizability of choice correspondences and choice axioms have been examined extensively in social choice theory. 
In a series of works, \citet{sen1970collective,sen1971choice,sen1977social} examined the rationalizability of choice correspondences in an abstract setting.\footnote{See \citet{cato2018choice} for a detailed analysis of the relationships among several choice axioms.} 
In the context of the bargaining problem, \citet{peters1991independence} directly investigated the implications of Nash's IIA in single-valued bargaining solutions. 
\citet{ok1999revealed} showed that \textit{Arrow axiom} plays a crucial role in the rationalizability of multi-valued bargaining solutions. 
\citet{nakamura2025wp,nakamura2025wp2} also considered a weakening of \textit{Arrow axiom}, but adopted a different approach. 
Nakamura’s axiom postulates the requirement of \textit{Arrow axiom} only when 
each player’s maximum utility level coincides in both the original and new problems
(for similar axioms, see also \citet{DUBRA2001131} and \citet{xu2006alternative}). 
Combined with standard axioms, especially \textit{scale invariance}, this weakened choice axiom implies rationalizability under the 0--1 normalization of each player’s utility levels.

It should be noted that we employ a domain different from \citeauthor{nash1950bargaining}'s (\citeyear{nash1950bargaining}) original formalization. 
In the most standard framework of the bargaining problem, feasible sets are assumed to be convex, and bargaining solutions are formalized as choice functions. 
By contrast, we do not assume the convexity of feasible sets and allow solutions to select multiple outcomes. 
The non-convex bargaining problem has been considered in several studies, such as \citet{kaneko1980extension}, \citet{herrero1989nash}, \citet{conley1991bargaining},
\citet{zhou1997nash},
\citet{ok1999revealed}, \citet{mariotti2000maximal}, \citet{xu2006alternative,xu2013rationality,xu2019equitable}, and \citet{peters2012wpo}.
For surveys, see also \citet{xu2020nonconvex} and \citet{Thomson2022RED}. 
The convexity has been justified through the assumptions that randomization among agreements is allowed and that each player follows the expected utility theory. 
In the non-convex problem, we can analyze the situations in which these assumptions do not hold. 
Note that in such cases, \textit{scale invariance} still has compelling interpretations, even though this axiom has traditionally been justified within the expected utility framework.
For a detailed discussion, see \citet{xu2006alternative}.

A coarse Nash solution can be regarded as the combination of maximizers of multiple choice criteria, as illustrated in the examples we discussed earlier. The idea of conjoining multiple choice criteria is closely related to \citet{tadenuma2002efficiency,tadenuma2005egalitarian} and \citet{houy2009lexicographic}, where the choice correspondence is defined through the sequential application of efficiency and fairness criteria in the context of general resource allocation problems.
See also  \citet{manzini2012choice} for a general model of the sequential composition of different choice criteria.
In contrast to this sequential approach, our choice correspondence takes into account the intersection of multiple criteria. Such intersections typically involve trade-offs, which often result in incompleteness and render the resulting predictions coarse (e.g., \citealp{sen1973economic,kamaga2018utilitarianism}).

Coarse Nash solutions are also related to incomplete preferences studied in decision theory, such as the expected multi-utility model \citep{dubra2004expected} and the multi-prior expected utility model \citep{bewley2002knightian}. 
These preferences can be represented by the unanimity rule among multiple parameters in the decision maker's mind, and are characterized by a violation of completeness. 
The special cases considered in Example \ref{exm:2_multiweight} correspond to these models. 
The class of coarse Nash solutions encompasses a wider range of rules, such as the solution in Example \ref{exm:3_union}. 
This difference arises from the properties of improving sets, which are not necessarily convex, whereas the upper contour sets of the above decision-making models are convex. 

Furthermore, Proposition \ref{prop_improving set} is related to \citeauthor{aumann1962utility}'s (\citeyear{aumann1962utility}) result. 
Aumann showed that for any risk preference that satisfies the axioms of the expected utility theory except for completeness, there exists a von Neumann-Morgenstern utility function compatible with the preference in a weak sense. 
Proposition \ref{prop_improving set} shows that a similar result holds for a binary relation induced by any improving set: 
There exists a weight over players that is compatible with the binary relation. 
Despite their similarities, our proposition is not a direct application of Aumann’s result or argument because of the differences in the geometric properties of the upper contour sets:  
Unlike Aumann’s argument, improving sets are not necessarily convex cones, and hence the separating hyperplane theorem cannot cover all the cases in our analysis.

\section{Concluding comments}\label{sec_con}

This paper has focused on the limitations of \textit{Arrow axiom} in the classical bargaining problem, particularly targeting its component, \textit{dual Chernoff axiom}. 
We have introduced two new rationality axioms that weaken \textit{Arrow axiom}: \textit{weak Arrow axiom} and \textit{IIE}. 
The main contribution of this paper is to show that these new axioms, when combined with standard axioms, characterize a new class of solutions called coarse Nash solutions.

Coarse Nash solutions generalize the dominance relation of the Nash solution by employing more general open sets, which we call  
improving sets. 
This class of solutions encompasses a wide range of outcomes, including weighted Nash solutions and solutions that assign all weakly Pareto optimal outcomes as extreme cases. 
As its name suggests, this solution always selects a larger (coarser) set of outcomes than some weighted Nash solution.
Consequently, it provides a more flexible guideline for collective decision-making in real-world bargaining situations where multiple criteria, such as efficiency and equity, may conflict, making it difficult to construct a single objective function.


\begin{center}
\Large{{\bf Appendix}}
\end{center}

\appendix

\section{Proof of Proposition \ref{prop_weakArrow}}

Suppose that $F$ satisfies \textit{Chernoff axiom} and \textit{weak dual Chernoff axiom}.
Take any $S, T\in\mathcal{B}$.
By \textit{weak dual Chernoff axiom}, $F(S) \cap F(T) \subseteq F(S \cup T)$.
Clearly, $F(S) \cap F(T) \subseteq S \cap T$.
Therefore, $F(S)\cap F(T) \subseteq  F(S \cup T) \cap S\cap T$ holds.
Then, suppose that $F(S \cup T) \cap S\cap T \neq \emptyset$.
Since $F(S\cup T)\cap S\neq \emptyset$ and $F(S\cup T)\cap T\neq \emptyset$, by \textit{Chernoff axiom}, we have $F(S\cup T)\cap S \subseteq F(S)$ and $F(S\cup T)\cap T \subseteq F(T)$.
Therefore, $F(S)\cap F(T) \supseteq F(S \cup T) \cap S\cap T$.

Conversely, suppose that $F$ satisfies the property of the proposition, \textit{weak Arrow axiom}.
Then, we show that $F$ also satisfies \textit{Chernoff axiom} and \textit{weak dual Chernoff axiom}.
For \textit{Chernoff axiom}, take $S, S'\in \mathcal{B}$ with $S\subseteq S'$.
By \textit{weak Arrow axiom}, we have $F(S)\cap F(S') = F(S') \cap S$. 
By $F(S)\cap F(S') \subseteq F(S)$, we have $F(S) \supseteq  F(S')\cap S$. 
For \textit{weak dual Chernoff axiom},  take $S, T \in \mathcal{B}$. 
Then, by \textit{weak Arrow axiom}, we have $F(S)\cap F(T) = F(S \cup T) \cap S\cap T \subseteq F(S \cup T) $. 
\qed

\section{Proof of Lemma \ref{prop_rationalization}}\label{appendix_proof}

Since the if part is clear in each point, we show the only-if part.
Given a bargaining solution $F$, define a binary relation $\succsim^F$ on $\mathbb{R}^n_{++}$ such that
\[
x \succsim^F y\iff x \in F(\text{cmp}\{x,y\})
\]
for any $x,y \in \mathbb{R}^n_{++}$.
\vspace{3mm}

\textbf{Step 1}: If $F$ satisfies \textit{efficiency}, then $\succsim^F$ is complete.

Take any $x,y \in \mathbb{R}^n_{++}$.
Since $F$ is nonempty for any $S \in \mathcal{B}$, there exists $z \in \text{cmp}\{x,y\}$ such that $z \in F(\text{cmp}\{x,y\})$.
By definition, $x \ge z$ or $y \ge z$.
Therefore, by \textit{efficiency}, we have $x \in F(\text{cmp}\{x,y\})$ or $y \in F(\text{cmp}\{x,y\})$, that is,  $x\succsim^F y$ or $y \succsim^F x$ holds.
\vspace{3mm}

\textbf{Step 2}: If $F$ satisfies \textit{efficiency} and \textit{continuity}, then $\succsim^F$ is monotone and continuous.

The monotonicity follows from \textit{efficiency}. 
To show the continuity, by Step 1, it suffices to prove that for all $x\in \mathbb{R}^n_{++}$, the set $U_x = \{ y\in \mathbb{R}^n_{++} \mid y \succ^F x \}$ is open, or equivalently, $U_x^c = \{ y\in \mathbb{R}^n_{++} \mid y \not\succ^F x \}=\{y \in \mathbb{R}^n_{++} \mid x \succsim^F y\}$ is closed.
Let $x, z\in  \mathbb{R}^n_{++}$ and $\{ z^k \}_{k\in \mathbb{N}} \subseteq U_x^c$ be a sequence that converges to $z$. 
By the definition of $\succ$, $x\in F(\text{cmp} \{ x, z^k\})$. 
By \textit{continuity}, $x \in F(\text{cmp} \{ x, z\})$, which means that $z \not\succ x$. 
Therefore, $U_x^c$ is a closed set.

\vspace{3mm}

\textbf{Step 3}: If $F$ satisfies \textit{weak Arrow axiom} and \textit{IIE}, then $\succsim^F$ is quasi-transitive.

Let $x,y,z\in \mathbb{R}^n_{++}$ be such that $x\succ^F y$ and $y\succ^F z$. 
By the definition of $\succ^F$, we have $x\in F(\text{cmp} \{ x,y \})$, $y\notin F(\text{cmp} \{ x,y \})$, $y\in F(\text{cmp} \{ y,z \})$, and $z\notin F(\text{cmp} \{ y,z \})$. 
We show that $x\in F(\text{cmp} \{ x,y, z \})$,  $y\notin F(\text{cmp} \{ x,y, z \})$, and $z\notin F(\text{cmp} \{ x,y, z \})$.
Note that since $x\in F(\text{cmp} \{ x,y, z \})$ follows from $y\notin F(\text{cmp} \{ x,y, z \})$, $z\notin F(\text{cmp} \{ x,y, z \})$, and \textit{efficiency}, it is sufficient to prove the latter two ones. 
Suppose to the contrary that $y\in F(\text{cmp} \{ x,y, z \})$ or $z\in F(\text{cmp} \{ x,y, z \})$. 
If $y\in F(\text{cmp} \{ x,y, z \})$, then by \textit{weak Arrow axiom},  $y\in  F(\text{cmp} \{ x,y, z \}) \cap \text{cmp} \{ x,y \} \cap \text{cmp} \{ x,y, z \} \subseteq F(\text{cmp} \{ x,y \})$. 
This contradicts $y\notin F(\text{cmp} \{ x,y \})$. On the other hand, if $z\in F(\text{cmp} \{ x,y, z \})$, then by \textit{weak Arrow axiom}, $z\in  F(\text{cmp} \{ x, y,z \}) \cap \text{cmp} \{ y, z \} \cap \text{cmp} \{ x, y, z \} \subseteq F(\text{cmp} \{ y,z \})$, which is also a contradiction with  $z\notin F(\text{cmp} \{ y,z \})$.

We then prove that $x\succ^F z$, that is, $x\in F(\text{cmp} \{ x,z \})$ and $z\notin F(\text{cmp} \{ x,z \})$. 
By $x\in F(\text{cmp} \{ x,y, z \})$ and \textit{weak Arrow axiom}, $x\in   F(\text{cmp} \{ x,y, z \}) \cap \text{cmp} \{ x, z \}\cap \text{cmp} \{ x, y, z \} \subseteq F(\text{cmp} \{ x,z \})$. 
To show $z\notin F(\text{cmp} \{ x,z \})$, suppose to the contrary that $z\in  F(\text{cmp} \{ x,z \})$. 
By $z\in  F(\text{cmp} \{ x,z \}) \backslash  F(\text{cmp} \{ x,y, z \})$ and the contraposition of \textit{IIE}, there exists $w\in F(\text{cmp} \{ x,y, z \}) \backslash \text{cmp} \{ x,z \}$. Therefore, $w\in \text{cmp} \{ y \}$. 
By \textit{efficiency}, $y\in F(\text{cmp} \{ x,y, z \})$, which is a contradiction to $y\notin F(\text{cmp} \{ x,y,z \})$.  
\vspace{3mm}

\textbf{Step 4}: If $F$ satisfies \textit{efficiency}, \textit{continuity}, \textit{weak Arrow axiom}, and \textit{IIE}, then it can be weakly rationalizable by $\succ^F$, that is, 
\[
F(S)=\{x \in S\mid \not \exists y \in S ~~\text{s.t.} ~~ y\succ^F x\},~~\forall S \in \mathcal{B}.
\]

Take $S \in \mathcal{B}$ arbitrarily. 
To prove $F(S) \subseteq  \{ x\in S\mid \nexists y\in S ~~\text{s.t.} ~~ y\succ^F x \}$, suppose to the contrary that there exist $x, y \in S$ with $x\in F(S)$ but $y\succ x$. 
By \textit{weak Arrow axiom}, $F(S\cup \text{cmp}  \{ x, y\})\cap S \cap \text{cmp} \{ x, y\}  =  F(S) \cap F(\text{cmp}  \{ x, y\}) $. Since $x$ is in the left hand side, $x\in F(\text{cmp}  \{ x, y\})$, a contradiction to $y\succ^F x$. 

To prove the converse, let $x\in S$ be such that there is no $y\in S$ with $y\succ^F x$. 
By the definition of $\succsim^F$, $x \in F(\text{cmp}\{x,y\})$ or $y \notin F(\text{cmp}\{x,y\})$ for all $y \in S \setminus \{x\}$.  
However, if $y \notin F(\text{cmp}\{x,y\})$ for some $y \in S \setminus \{x\}$, by \textit{efficiency}, we must have $x \in F(\text{cmp}\{x,y\})$.
Hence, 
\begin{equation}
\label{eq:pr_weakrat1}
x \in F(\text{cmp}\{x,y\}),~~\forall y \in S. 
\end{equation} 
Let $\{ y^k \}_{k\in \mathbb{N}}\subseteq S$ be a sequence such that $\{ \text{cmp} \{x,  y^1 , \ldots , y^k \} \}_{k\in \mathbb{N}} \subseteq \mathcal{B}$ converges to $S$. 
For all $K\in \mathbb{N}$, if $x\in F(\text{cmp} \{x,  y^1 , \ldots , y^K \})$, then by \textit{weak Arrow axiom} and \eqref{eq:pr_weakrat1}, we have  $x\in F(\text{cmp} \{x,  y^1 , \ldots , y^K \}) \cap F(\text{cmp} \{x,  y^{K+1} \}) \subseteq  F(\text{cmp} \{x,  y^1 , \ldots , y^K, y^{K+1} \})$.  
This observation implies that for all $k\in \mathbb{N}$, $x\in F(\text{cmp} \{x,  y^1 , \ldots , y^k \})$. By \textit{continuity}, we have $x\in F(S)$. 
\vspace{3mm}

\textbf{Step 5}: If $F$ satisfies \textit{efficiency}, and \textit{Arrow axiom}, then $\succsim^F$ is transitive.

Let $x, y,z \in \mathbb{R}^n_{++}$ be such that $x \succsim^F y$ and $y \succsim^F z$.
By \textit{efficiency}, $w \in F(\text{cmp}\{x,y,z\})$ for at least one $w \in \{x,y,z\}$.
First, we show that $x \in F(\text{cmp}\{x,y,z\})$.
Suppose that $y \in F(\text{cmp}\{x,y,z\})$.
Then, by \textit{Arrow axiom}, $F(\text{cmp}\{x,y\})=F(\text{cmp}\{x,y,z\}) \cap \text{cmp}\{x,y\}$.
Since $x \succsim^Fy$, we have $x \in F(\text{cmp}\{x,y\})$, so $x \in F(\text{cmp}\{x,y,z\})$.
Suppose that $z \in F(\text{cmp}\{x,y,z\})$.
Similarly, by \textit{Arrow axiom}, $F(\text{cmp}\{y,z\})=F(\text{cmp}\{x,y,z\}) \cap \text{cmp}\{y,z\}$.
Since $y \succsim^Fz$, we have $y \in F(\text{cmp}\{y,z\})$, so $y \in F(\text{cmp}\{x,y,z\})$.
Then, as the above argument, we can also see that $x \in F(\text{cmp}\{x,y,z\})$.
Finally,  by \textit{Arrow axiom}, we can see that $x \in F(\text{cmp}\{x,y,z\})\cap \{x,z\}=F(\text{cmp}\{x,z\})$, that is, $x \succsim^F z$.
\vspace{3mm}

By Steps 1--5, we can obtain the result as follows.
First, Steps 2, 3 and 4 directly imply that (2) of the proposition holds.
Moreover, if $F$ satisfies \textit{Arrow axiom}, by Steps 1 and 4,
\[
F(S)=\{x \in S\mid \not \exists y \in S ~~\text{s.t.} ~~ y\succ^F x\}=\{x \in S \mid x \succsim^F y, \forall y \in S\},~~\forall S \in \mathcal{B}.
\]
Then, by Steps 1 and 5, $\succsim^F$ is complete and transitive, which is the assertion of (1).
\qed

\section{Proof of Theorem \ref{thm:main}}

We first show the only-if part. 
Let $F$ be a solution that satisfies  \textit{weak Arrow axiom}, \textit{IIE}, \textit{efficiency},  \textit{scale invariance}, 
and \textit{continuity}. 
By (2) of Lemma \ref{prop_rationalization}, we have known that
\[
F(S) = \{ x\in S\mid \nexists y\in S ~~\text{s.t.} ~~ y\succ x \},~~\forall S \in \mathcal{B},
\]
where $\succ$ is the asymmetric part of the quasi-transitive, monotone, and continuous binary relation $\succsim$ defined as
\begin{equation*}
        x\succsim y \iff  x\in F(\text{cmp}  \{ x, y\}) 
    \end{equation*}
for all $x,y\in \mathbb{R}^n_{++}$.
Moreover, by \textit{scale invariance}, 
it is easy to verify that $\succ$ is 
scale invariant (i.e., for all $a, x,y\in \mathbb{R}^n_{++}$, $x\succ y$ if and only if $a * x \succ a * y$). 
Furthermore, by the definition of $\succ$, it is asymmetric. 

Next, we construct an improving set.
For $x\in \mathbb{R}^n_{++}$, define $A_x \subseteq \mathbb{R}^n$ as 
\begin{equation}
    \label{eq:def_Ax}
        A_x =\{ (\log y_1 , \ldots, \log y_n ) - (\log x_1 , \ldots, \log x_n )\in \mathbb{R}^n \mid  y\in \mathbb{R}_{++}^n ~~\text{s.t.} ~~ y\succ x \}. 
\end{equation}
We  prove that for all $x, x' \in \mathbb{R}^n_{++}$, $A_x = A_{x'}$. 
For $z\in A_x$, there exists $y\in \mathbb{R}^n_{++}$ such that $y\succ x$ and $z = (\log y_1 , \ldots, \log y_n ) - (\log x_1 , \ldots, \log x_n )$. Take $a\in \mathbb{R}^n_{++}$ as $a\ast x = x'$. 
Since $\succ$ is scale invariant, $a\ast y\succ a\ast x = x'$. 
By 
    \begin{align*}
        &(\log a_ 1y_1 , \ldots, \log a_n y_n ) - (\log a_1 x_1 , \ldots, \log a_n x_n ) \\
        &= 
        (\log y_1  + \log a_1, \ldots, \log y_n + \log a_n ) - (\log x_1 + \log a_1,, \ldots, \log  x_n  + \log a_n) \\
        &= (\log y_1 , \ldots, \log y_n ) - (\log x_1 , \ldots, \log x_n ) \\
        &= z, 
    \end{align*}
we have $z\in A_{x'}$, which implies $A_x\subseteq A_{x'}$. By a similar argument,
we can prove the converse inclusion. 
Therefore, $A_x$ does not depend on $x$. 
Henceforth, we denote the set in \eqref{eq:def_Ax} by $A$. 

We then show that the set $A$ is an improving set.
Since $U_x = \{ y\in \mathbb{R}^n_{++} \mid y\succ x\}$
is open and the logarithm function is continuous, $A$ is also open. 
The property (i) follows from the monotonicity and asymmetry of $\succ$ and the monotonicity of the logarithm function. 
Then, we verify the property (ii).  
Suppose to the contrary that there exist $z, z'\in A$ with $z+ z'\notin A$.  
Let $x, y, w\in \mathbb{R}^n_{++}$ be such that $z =   (\log x_1 , \ldots, \log x_n ) - (\log y_1 , \ldots, \log y_n )$ and $z' =  (\log y_1 , \ldots, \log y_n ) - (\log w_1 , \ldots, \log w_n )$. 
By $z, z'\in A$, we have $x\succ y$ and $y \succ w$. 
However, $z+ z' = (\log x_1 , \ldots, \log x_n ) - (\log w_1 , \ldots, \log w_n ) \notin A$ implies that $x\not\succ w$, which is a contradiction to  the quasi-transitivity of $\succsim$. 
    
Finally, by the above discussion, we have shown that $y \succ x $ is equivalent to $(\log y_1 , \ldots,  \log y_n) - (\log x_1 , \ldots,  \log x_n) \in A$ for any $x, y \in \mathbb{R}^n_{++}$ and $A$ is an improving set.
Therefore, we can conclude that, for all $S \in \mathcal{B}$,
\begin{eqnarray*}
F(S)&=&\{ x\in S\mid \nexists y\in S ~~\text{s.t.} ~~ y\succ x \}\\
&=&\{ x\in S \mid \nexists y \in S  ~~ \text{s.t. } (\log y_1 , \ldots,  \log y_n) - (\log x_1 , \ldots,  \log x_n) \in A \},
\end{eqnarray*}
which completes the proof of the only-if part.

\vspace{3mm}

    Next, we prove the if part. It is easy to prove that the coarse Nash solutions satisfy \textit{efficiency}. 
    To verify the remaining axioms, let $F^A$ be some coarse Nash solution. 
    \begin{itemize}
        \item \textit{Scale invariance}: Let $S\in\mathcal{B}$, $a\in \mathbb{R}^n_{++}$, and $x\in F^A(S)$. 
        By the definition of $F^A$, $F^A (S) = \{ x\in S \mid \nexists y \in S  ~~ \text{s.t. } (\log y_1 , \ldots,  \log y_n) - (\log x_1 , \ldots,  \log x_n) \in A \}$. Therefore, by the definition of $a\ast S$, 
        \begin{align*}
            F^A (a\ast S) 
            &= \{ x\in S \mid \nexists y \in S  ~~ \text{s.t. } (\log a_1 y_1 , \ldots,  \log a_n y_n) - (\log a_1 x_1 , \ldots,  \log a_n  x_n) \in A \} \\
            &= \{ x\in S \mid \nexists y \in S  ~~ \text{s.t. } (\log  y_1 , \ldots,  \log y_n) - (\log  x_1 , \ldots,  \log x_n) \in A \} \\
            &= F^A (S). 
        \end{align*}
        \item \textit{Continuity}: Let  $\{S^k \}_{k\in \mathbb{N}} \subseteq \mathcal{B}$ and $\{ x^k\}_{k\in \mathbb{N}}\subseteq \mathbb{R}_{++}^n$ be such that (i)  $x^k \in F^A(S^k)$ for all $k\in \mathbb{N}$; (ii) $\{S^k \}_{k\in \mathbb{N}}$ converges to some $S \in \mathcal{B}$; (iii) $\{ x^k\}_{k\in \mathbb{N}}$ converges to some $x \in  S$ in $\mathbb{R}_{++}^n$. 
        Take $z\in S$ arbitrarily. Let $\{z^k\}_{k\in\mathbb{N}} \subseteq \mathbb{R}^n_{++}$ be such that $z^k \in S^k$ for each $k\in \mathbb{N}$ and $\{z^k\}_{k\in\mathbb{N}}$ converges to $z$. 
        By the definitions of $F^A$ and $\{x^k\}_{k\in\mathbb{N}}$, we have $( \log z^k_1, \ldots, \log z^k_n)  - (\log x^k_1, \ldots, \log x^k_n) \notin A$ for each $k\in \mathbb{N}$. 
        By the openness of $A$ and the continuity of the log function, $ ( \log z_1, \ldots, \log z_n)  - (\log x_1, \ldots, \log x_n)\notin A$ holds. 
        Since this holds for each $z\in A$, we have $x\in F^A(S)$. 
        
        \item \textit{Weak Arrow axiom}: Let $S, T \in \mathcal{B}$ and $x\in \mathbb{R}^n_{++}$.  
        If $x\in F^A(S) \cap F^A(T)$, then for all $y\in S\cup T$, $ (\log  y_1 , \ldots,  \log y_n) - (\log  x_1 , \ldots,  \log x_n) \not\in A$, which implies $x\in F^A(S\cup T)$. For the converse, suppose that $x\in F^A(S\cup T) \cap S \cap T$. Then, for all $y\in S\cup T$, $ (\log  y_1 , \ldots,  \log y_n) - (\log  x_1 , \ldots,  \log x_n) \not\in A$, which implies $x\in F^A(S)\cap F^A(T)$. 

        \item \textit{IIE}: Let $S, S' \in \mathcal{B}$ be such that $S\subseteq S'$ and $F^A(S') \subseteq S$. 
        Take $x\in F^A (S)$ arbitrarily.
        By the definition of $F^A$,  
        \begin{equation}
        \label{eq:if_wssa}
            (\log  y_1 , \ldots,  \log y_n) - (\log  x_1 , \ldots,  \log x_n) \not\in A,~~\forall y\in S.
        \end{equation} 
        If $x\notin F^A (S')$, then there exists $z\in S' \setminus S$ such that  $ (\log  z_1 , \ldots,  \log z_n) - (\log  x_1 , \ldots,  \log x_n) \in A$. 
        We can set that $z$ is in $F^A (S')$ by the property (ii) of improving sets and \eqref{eq:if_wssa}. This is a contradiction to $F^A (S')\subseteq S $.  
        
    \end{itemize}
\qed

\bibliography{reference.bib}

\end{document}